%% file: main.tex
\documentclass[sigconf, screen, dvipsnames, nonacm]{acmart}

\usepackage{hyperref}
\usepackage[noend]{algorithmic}
\usepackage{algorithm}
\usepackage{multirow}
\usepackage{mathtools}
\usepackage{amsmath}
\usepackage{bbm}
\usepackage{booktabs}
\usepackage[inline]{enumitem} 
\usepackage{subcaption}
\usepackage{bm} 
\usepackage{xcolor}
\usepackage{balance}
\usepackage{lipsum}

\usepackage{wrapfig}
\usepackage{siunitx}

\usepackage{hwemoji}

\usepackage{arydshln}
\makeatletter
\def\adl@drawiv#1#2#3{
        \hskip.5\tabcolsep
        \xleaders#3{#2.5\@tempdimb #1{1}#2.5\@tempdimb}%
                #2\z@ plus1fil minus1fil\relax
        \hskip.5\tabcolsep}
\newcommand{\cdashlinelr}[1]{%
  \noalign{\vskip\aboverulesep
          \global\let\@dashdrawstore\adl@draw
          \global\let\ adl@draw\adl@drawiv}
  \cdashline{#1}
  \noalign{\global\let\adl@draw\@dashdrawstore
          \vskip\belowrulesep}}
\makeatother

\usepackage{tikz}
\usetikzlibrary{automata, arrows, bayesnet, bending}

\usepackage{tikzscale}

\newtheorem{proposition}{Proposition}

\newtheorem{definition}{Definition}

\newtheorem{remark}{Remark}

\copyrightyear{2026}
\acmYear{2026}
\setcopyright{none} 
\acmConference[CIKM '26]{Proceedings of the 35th ACM International Conference on Information and Knowledge Management}{October 2026}{Rome, Italy}
\acmBooktitle{Proceedings of the 35th ACM International Conference on Information and Knowledge Management (CIKM '26), October 2026, Rome, Italy}
\acmDOI{}
\acmISBN{}

\begin{document}
\title[Distributional Approximate Nearest Neighbour Search for Uncertainty-Aware Retrieval]{Distributional Approximate Nearest Neighbour~\\
Search for Uncertainty-Aware Retrieval}
\subtitle{\textsc{dinosaur}  🦕}

\author{Olivier Jeunen}
\affiliation{
  \institution{aampe}
  \city{Antwerp}
  \country{Belgium}
}

\begin{abstract}
Approximate Nearest Neighbour search indices form the backbone of real-world recommender systems, enabling real-time candidate retrieval over million-item catalogues.
Typically, a single point estimate embedding is learnt for every user and every item.
At serving time, the user embedding queries the index for relevant items.
Since these representations are learnt from sparse interaction data, they are noisy and might fail to capture all the nuances that contribute to ``\textit{relevance}''---ignoring the fundamental uncertainty that is inherent to them.
The result is a retrieval pipeline that is systematically biased toward the small minority of popular head items with well-estimated embeddings, at the expense of the long-tail majority of niche, diverse, and serendipitous content.

We propose \textsc{dinosaur} (\textbf{Di}stributional Approximate Nearest \textbf{N}eighb\textbf{o}ur \textbf{S}e\textbf{a}rch for \textbf{U}ncertainty-Aware \textbf{R}etrieval): a simple and infrastructure-compatible framework to incorporate embedding uncertainty into candidate generation.
Rather than indexing point estimates, \textsc{dinosaur} samples $S_i$ embeddings per item and constructs an index on this augmented set.
Analogously, at query time, a user embedding is sampled.
This two-sided stochastic retrieval process implicitly marginalises over embedding uncertainty, without requiring changes to model architecture or ANN index infrastructure.

On the analytical side, we show that \textsc{dinosaur} recovers standard point-estimate retrieval as uncertainty vanishes, and we characterise how increased embedding variance expands the regions of latent space in which uncertain items are retrievable.
Reproducible empirical observations align with these expectations, showing large coverage gains with small losses in offline recall.
\end{abstract}

\maketitle

\input{1.Introduction}

\input{2.Methodology}
\input{3.Experiments}
\input{4.Conclusions}

\balance

\bibliographystyle{ACM-Reference-Format}
\bibliography{bibliography}

\end{document}

%% file: 1.Introduction.tex
\section{Introduction \& Motivation}\label{sec:intro}
Modern large-scale recommender systems typically consist of multiple sequential stages that serve to refine iteratively smaller sets of items, down to an item or slate of items to be shown to an end-user~\cite{Covington2016,Liu2017,Higley2022}.
The first stage is often referred to as \textit{candidate generation}, filtering $\mathcal{O}(10^6)$ or more items down to $\mathcal{O}(10^3)$~\cite{Asadi2013,Jeunen2024_FIRE}.

Practical implementations of this paradigm typically leverage Approximate Nearest Neighbour (ANN) search: users and items are represented by $d$-dimensional embeddings in a shared vector space, and a pre-computed index over items is queried with the user embedding at serving time.
When similarity is measured by inner product or cosine, this reduces to a Maximum Inner Product Search (MIPS) problem---a well-known database problem for which several efficient production-grade open-source solutions exist~\cite{Aumuller2020} (e.g. \textsc{annoy}~\cite{Github:annoy}, \textsc{faiss}~\cite{Douze2026}, \textsc{scann}~\cite{Guo2020_SCANN}, or \textsc{hnsw}~\cite{Malkov2020}).

This pipeline is efficient and scalable, but it rests on a strong and rarely scrutinised assumption: that a single point estimate faithfully represents the user or item.
In reality, both representations are fundamentally uncertain~\cite{Salakhutdinov2008}.
Some uncertainty is epistemic---learning from sparse interaction data implies high variance for long-tail content; other uncertainty is aleatoric---the low-rank assumption that justifies a latent factor model necessarily loses information~\cite{Hullermeier2021}.
The consequence is a mismatch between a static ANN setup on one hand, and users' diverse and context-dependent information needs paired with highly uncertain item embeddings on the other.

This leads to a systematic bias in retrieval.
Point-estimate ANN search is deterministic by design.
As a result, items that fall just outside of top-$k$ neighbourhoods are permanently excluded from candidate generation, regardless of underlying uncertainty in the item or query.
This produces well-documented pathologies in deployed systems: popularity bias~\cite{Abdollahpouri2019} potentially amplified by feedback loops~\cite{Chaney2018}, and general disparities in exposure that systematically starve tail creators and sellers of impressions~\cite{Diaz2020,Jeunen2021_TS}.
The value of \emph{exploration}, on the other hand, is clear, multi-sided, and sustained~\cite{Guo2020,Su2024_exploration,Jeunen2026_umap}.

Existing remedies (e.g. diversity-promoting objectives~\cite{Liu2022,Liu2025}, post-hoc re-ranking~\cite{Kumar2023}, or exploration bonuses~\cite{Zhu2023,Chen2021}) are typically heuristic, applied at a later ranking stage in the pipeline, and disconnected from the underlying source of the problem.
We argue that the right place to address this is at the retrieval stage itself, by replacing point-estimate embeddings with \emph{distributional} embeddings that explicitly represent uncertainty.
Adapting existing ANN search frameworks to index and query \emph{distributions} over embeddings, however, is no small undertaking~\cite{Github:annoy,Douze2026,Guo2020_SCANN,Malkov2020}.

Our key insight is simple: if we sample multiple embedding vectors per item from a learnt distribution $\mathsf{P}(\mathbf{v}_i)$ and build the ANN index over these samples, items with higher embedding uncertainty naturally cover a larger region of embedding space---and are therefore more likely to be retrieved for users whose true preferences lie in that region.
This requires no changes to the scoring model, no new hyper-parameters beyond the number of samples $S_i$, and no modifications to the ANN index infrastructure itself.
Furthermore, this mechanism provides a principled route to retrieval-stage exploration for long-tail content, with strong connections to Thompson Sampling from the existing literature on sequential decision-making and the explore-exploit trade-off~\cite{Thompson1933}.

We formalise this idea as \textsc{dinosaur} (\textbf{Di}stributional Approximate Nearest \textbf{N}eighb\textbf{o}ur \textbf{S}e\textbf{a}rch for \textbf{U}ncertainty-Aware \textbf{R}etrieval), provide a rigorous theoretical analysis of its properties, and validate it empirically with a reproducible setup.

%% file: 2.Methodology.tex
\section{Problem Setting and Notation}
\label{sec:setting}

Let $\mathcal{U}$ and $\mathcal{I}$ denote the sets of users (e.g. buyers or subscribers) and items (e.g. sellers or creators) respectively, with $|\mathcal{I}| = M$. 
In standard dual-encoder retrieval, representations are treated as deterministic points. Instead, we assume each item $i \in \mathcal{I}$ is associated with a continuous embedding distribution $\mathsf{P}(\mathbf{v}_i)$ over $\mathbb{R}^d$, characterized by mean $\mu_i$ and covariance $\Sigma_i$. 
Similarly, each user $u \in \mathcal{U}$ possesses an embedding distribution $\mathsf{P}(\mathbf{v}_u)$. We denote draws from these distributions as $\mathbf{v}_i \sim \mathsf{P}(\mathbf{v}_i)$ and $\mathbf{v}_u \sim \mathsf{P}(\mathbf{v}_u)$.

Let $\mathrm{sim} : \mathbb{R}^d \times \mathbb{R}^d \to \mathbb{R}$ denote a standard similarity function (e.g. the inner product). 
Standard Approximate Nearest Neighbour (ANN) retrieval acts as a deterministic truncation, returning the top-$k$ items ranked strictly by $\mathrm{sim}(\mu_u, \mu_i)$. This rigid truncation disproportionately penalises highly uncertain, long-tail items whose true relevance may be understated by their point-estimate means.

\begin{definition}[\textsc{dinosaur} Index and Retrieval]
\label{def:dinosaur}
Given a sample budget $S_i \geq 1$ per item, the \emph{\textsc{dinosaur} index} $\mathcal{X}_S$ is constructed by drawing independent samples 
$\mathbf{v}_i^{(1)}, \ldots, \mathbf{v}_i^{(S_i)} \overset{\text{iid}}{\sim} \mathsf{P}(\mathbf{v}_i)$ 
for each item $i \in \mathcal{I}$, and indexing the resulting vectors:
$$ \mathcal{X}_S = \bigcup_{i \in \mathcal{I}} \left\{ \left(\mathbf{v}_i^{(j)}, i\right) : j = 1, \ldots, S_i \right\}. $$
At query time, a user embedding $\mathbf{v}_u \sim \mathsf{P}(\mathbf{v}_u)$ is sampled. Retrieval consists of finding the nearest vectors in $\mathcal{X}_S$ to $\mathbf{v}_u$ and returning their corresponding item identifiers, deduplicating until exactly $k$ distinct items are accumulated.
\end{definition}

\begin{definition}[Supplier Exposure and Catalogue Coverage]
\label{def:retrievability}
The \emph{expected exposure} of item $i$ for user $u$ is its retrieval probability: $\rho(i, u) = \mathsf{P}\left( i \in \mathrm{Top}\text{-}k(u) \right)$. Under standard point-estimate retrieval, this collapses to a deterministic $\{0,1\}$. The \emph{catalogue coverage} of a retrieval scheme is the expected fraction of the total item catalogue that is retrievable for at least one user, serving as a proxy for marketplace health and tail-seller fairness.
\end{definition}


\section{The \textsc{dinosaur} Framework}
\label{sec:method}

Deploying \textsc{dinosaur} in practice requires only two ingredients: a generative source of embedding uncertainty and an existing, unmodified ANN index (e.g. \textsc{faiss}, \textsc{scann}, \textsc{hnsw}). 
Crucially, the method is entirely agnostic to the source of the uncertainty itself.
In practice, $\mathsf{P}(\mathbf{v}_i)$ can be obtained in several standard ways: from Bayesian latent-factor models~\cite{Salakhutdinov2008, Sakhi2020}, approximate Bayesian inference such as variational Bayes~\cite{Kingma2014_VAE,Blundell2015}, inference-time stochasticity such as Monte Carlo dropout~\cite{Gal2016}, or lightweight empirical proxies that tie posterior scale to interaction scarcity, e.g. $\sigma_i \propto (1+n_i)^{-\gamma}$, mirroring the dependence of uncertainty on sample size.

By sampling $S_i$ vectors per item offline, the serving-time path remains identical to standard deterministic ANN retrieval, up to the larger index and the deduplication of neighbours into items.

\textit{One-Sided vs. Two-Sided Exploration.}
\textsc{dinosaur} is highly modular; it can model uncertainty exclusively on the item side (suppliers or creators), the user side (buyers or subscribers), or both simultaneously. In dynamic two-sided marketplaces where user latency constraints prohibit real-time sampling, a \emph{one-sided (item-only)} implementation is sufficient to unilaterally drive exploration toward long-tail and niche content, unlocking the core coverage and exploration benefits we establish below.


\section{Theoretical Analysis}
\label{sec:theory}

We establish the core geometric and distributional properties of \textsc{dinosaur}. For clarity of notation, we assume isotropic Gaussian distributions $\mathsf{P}(\mathbf{v}_i) = \mathcal{N}(\mu_i, \sigma_i^2 \mathbf{I})$ evaluated via inner product.

\begin{proposition}[Safety and Recovery of Point-Estimate Retrieval]
\label{prop:recovery}
For any fixed sample budget $S_i \geq 1$, as embedding variance vanishes ($\sigma_i^2, \sigma_u^2 \to 0$), the probability that \textsc{dinosaur} retrieves the exact deterministic top-$k$ approaches $1$. 
\end{proposition}

\begin{proof}[Proof Sketch]
By applying a Gaussian tail bound to the sampled inner products $\langle \mathbf{v}_u, \mathbf{v}_i^{(j)} \rangle$, we can show that as $\sigma^2 \to 0$, the maximum sampled score converges strictly to the deterministic mean $\langle \mu_u, \mu_i \rangle$. Consequently, the rank order remains unviolated with high probability.
\end{proof}

\begin{proposition}[Monotonic Coverage Expansion for the Long-Tail]
\label{prop:coverage}
For any user $u$ and any item $i \notin \mathrm{Top}\text{-}k_{\mathrm{point}}(u)$, the retrieval probability $\rho_{\mathrm{\textsc{dinosaur}}}(i, u)$ is strictly non-decreasing with respect to its embedding variance $\sigma_i^2$. 
\end{proposition}

\begin{proof}[Proof Sketch]
Higher variance $\sigma_i^2$ dictates a convex geometric ordering, dispersing the item's probability mass across a broader volume of the latent space. This strictly increases the probability that at least one of the $S_i$ samples intersects the query's retrieval boundary, mathematically guaranteeing higher exposure for high-variance (i.e. highly uncertain) items.
\end{proof}

\begin{remark}[Marketplace Implications]
In most recommendation architectures, embedding variance is inversely proportional to historical interaction volume (i.e. $\sigma_i^2$ is highest for long-tail items). Proposition~\ref{prop:coverage} confirms that \textsc{dinosaur} is not merely adding random noise; it structurally acts as a retrieval-phase exploration bonus.
Niche sellers with high epistemic uncertainty secure a disproportionate, mathematically guaranteed boost in marketplace exposure, bypassing the need for explicit randomisation or post-hoc re-ranking.
\end{remark}


\section{Practical Instantiations \& Systems Trade-offs}
\label{sec:practical}

Deploying \textsc{dinosaur} in production requires navigating hardware constraints and the fundamental tension between user-focused short-term relevance, and item-side coverage and exploration.

\subsection{Adaptive Sample Budgets and Index Bloat}
Na\"ively indexing $S$ samples for an entire industry-scale catalogue ($M \sim 10^8$) introduces a prohibitive memory footprint.
Exploration utility, nevertheless, will be concentrated in the long-tail. 
We propose to resolve this via \emph{Adaptive Sampling}: the sample budget $S_i$ is dynamically conditioned on the item's uncertainty.
By allocating $S_i = 1$ to e.g. the top 20\% of head items, whose embeddings are typically better estimated, and reserving $S_i \in \{2, 5, 10\}$ exclusively for uncertain tail items, the index size scales fractionally rather than multiplicatively, effectively eliminating index bloat.

Search cost scales with the number of indexed samples that can participate in retrieval, but not necessarily linearly in practice: graph- and partition-based ANN methods typically inspect only a small subset of the augmented index at query time.
The main explicit overhead in our implementation is retrieving a larger raw candidate set followed by deduplication to $K$ item identifiers.
\textsc{dinosaur} therefore trades index size and raw-candidate retrieval depth for coverage; adaptive sampling provides a natural way to manage this trade-off when memory or latency is limited.

\subsection{Thompson Sampling Epistemic Posteriors}
\label{subsec:instantiation}
In settings where the underlying representation model is inherently deterministic (e.g. standard iALS~\cite{Rendle2022}), we can construct a mathematically motivated proxy for the posterior by isolating \emph{epistemic uncertainty}---the uncertainty arising strictly from a lack of historical observation. 
Drawing on standard Bayesian estimation principles, the posterior standard deviation of a learned parameter typically decays relative to the number of observations $n$. We mirror this dynamic directly in the geometric space by defining an empirical standard deviation for each item as $\sigma_i = \alpha / (1 + n_i)^{\gamma}$, where $n_i$ is the interaction count. By indexing samples drawn from this synthetic posterior, \textsc{dinosaur} effectively executes a highly scalable form of Thompson Sampling within the candidate generation stage. Sparse tail creators inherently exhibit broad sampling distributions, granting them the geometric variance necessary to explore and compete for retrieval, while data-rich blockbusters gracefully collapse toward their deterministic point-estimates.

\subsection{The Precision-Exploration Trade-off}
\label{subsec:tradeoff}
By expanding the spatial footprint of uncertain items, \textsc{dinosaur} inherently trades a marginal degree of top-level point-estimate precision to maximise exposure for the long-tail.
However, the severity of this perceived drop in offline recall is largely an artefact of the Missing Not At Random (MNAR) nature of historical recommendation datasets~\cite{Jeunen2018,Jeunen2019,Jadidinejad2021}.

Offline logs are heavily biased by the deterministic retrieval policies that collected them, trapping the system in an algorithmic feedback loop~\cite{Chaney2018, Jeunen2019REVEAL_EVAL}.
Consequently, when \textsc{dinosaur} successfully retrieves a highly relevant but historically unexposed tail item, offline metrics routinely penalise it as a ``miss.''
By actively injecting uncertainty, \textsc{dinosaur} breaks this feedback loop.
Furthermore, this dynamic is structurally advantageous in modern multi-stage architectures: the primary objective of ANN retrieval is high-recall candidate generation, not strict precision.
\textsc{dinosaur} effectively shifts the burden of relevance filtering away from the brittle geometric boundaries of the vector index and onto downstream ranking models, which possess the requisite feature depth to accurately assess these newly surfaced exploratory candidates.

%% file: 3.Experiments.tex
\section{Experiments \& Empirical Validation}
\label{sec:experiments}

To empirically validate the effectiveness of the \textsc{dinosaur} framework, we conduct fully reproducible evaluations on a large-scale, real-world recommendation dataset. Our experiments are designed to answer two primary research questions:
\begin{description}
    \item[\textbf{RQ1 (Exploration):}]
    \textit{Does distributional retrieval expand catalogue coverage over traditional methods under varying values of $K$?}
    \item[\textbf{RQ2 (Utility):}]
    \textit{Does the injection of stochastic uncertainty degrade the recall of the top-$K$ recommendations?}
\end{description}
We evaluate \textsc{dinosaur} against exact and approximate nearest neighbour baselines on the MovieLens-32M dataset~\cite{Harper2015}, using iALS embeddings of dimension $d=128$~\cite{Rendle2022}. The task is candidate generation: for each user, we retrieve a set of $K$ items from a corpus of 84,428 items, and evaluate whether the held-out test item appears in the retrieved set (Recall@$K$).
All code is written in \textsc{python} using \textsc{faiss}.

\textit{Data.} We use the full ML-32M dataset comprising approximately 32 million interactions between 200,948 users and 84,428 items. For evaluation, we adopt a leave-one-out protocol with a fixed random seed, holding out one randomly selected interaction per user as the test item. All remaining interactions form the training set.

\textit{Embeddings.} User and item embeddings are learned via iALS with regularization $\lambda = 0.01$, unobserved weight $\alpha_0 = 0.1$, embedding dimension $d = 128$, and 16 training iterations. Embeddings are used without normalisation, preserving the inner product as the model score.
We leverage a synthetic epistemic posterior based on item interaction counts rather than learned distributional embeddings. This isolates the retrieval mechanism from the representation-learning problem; learning calibrated embedding distributions end-to-end is a promising area for future work.

\textit{Methods.} We compare three retrieval strategies:
\begin{enumerate}
    \item \textbf{Flat}: Exact inner product search (\textsc{faiss}' \texttt{IndexFlatIP}~\cite{Douze2026}), serving as the upper bound on recall.
    \item \textbf{IVFFlat}: Approximate search using an inverted file index with $\texttt{nlist} = 4\sqrt{|\mathcal{I}|}$ clusters and $\texttt{nprobe} = 0.1 \cdot \texttt{nlist}$, no quantization.
    \item \textbf{\textsc{Dinosaur} ($S$)}: Each item is represented by $S$ embedding copies in the index---one at the learned mean and $S-1$ stochastic perturbations drawn from unit-direction noise scaled by $\sigma_i = \alpha / (1 + n_i)^{0.25}$ with $\alpha=0.1$, where $n_i$ is the interaction count of item $i$ and $\alpha$ is a base noise scale. The index is an IVFFlat index over $S \cdot |\mathcal{I}|$ vectors with proportionally scaled $\texttt{nlist}$ and the same $\texttt{nprobe}$ fraction. Retrieval of $k' = S \cdot K$ candidates is followed by deduplication to $K$ items, preserving score ordering.
    This allows us to isolate any effects on retrieval utility.
    We focus exclusively on item-side uncertainty (fixing $S_u=1$) to isolate the impact on catalogue coverage.
\end{enumerate}
\textit{Metrics.} We report Recall@$K$ and Catalogue Coverage@$K$ for $K \in \{10, 50, 100, 250, 500, 1000\}$. Catalogue Coverage measures the fraction of the item catalogue that appears in at least one user's top-$K$ list.
We additionally provide a popularity-stratified view: items are partitioned into buckets (quintiles) by rank (B0 = most popular, B4 = least popular), and we report stratified Catalogue Coverage and Inclusion Rate at $K=1000$.
Inclusion Rate is defined as the fraction of users that receives at least one retrieved item from the given popularity quintile.

\subsection{Results and Discussion}

\begin{figure*}[t]
\vspace{-3ex}
    \centering
    \includegraphics[width=\textwidth]{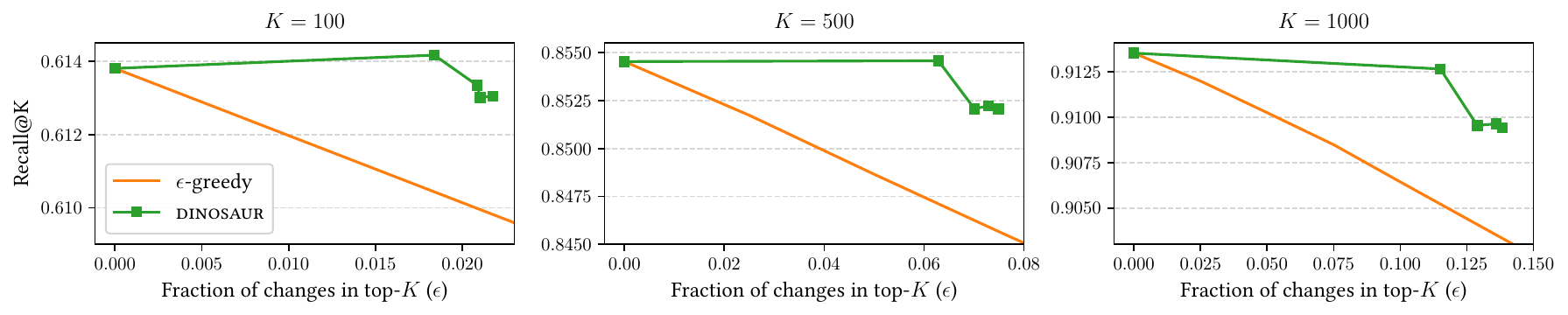}
    \caption{Pareto Frontier of Recall vs. Exploration at varying retrieval depths ($K \in \{100, 500, 1000\}$). \textsc{dinosaur} strictly dominates the $\epsilon$-greedy baseline, yielding substantially higher recall for an equivalent fraction of exploratory recommendations.}
    \label{fig:pareto}
\end{figure*}

\begin{table*}[t]
\vspace{-1ex}
\centering
\setlength{\tabcolsep}{4pt} 
\caption{Global Recall@$K$ and Catalogue Coverage@$K$ (\%) across retrieval methods. Higher is better.}
\label{tab:recall_coverage}
\begin{tabular}{l SSSSSS SSSSSS}
\toprule
& \multicolumn{6}{c}{\textbf{Recall @ $K$}} & \multicolumn{6}{c}{\textbf{Catalogue Coverage (\%) @ $K$}} \\
\cmidrule(lr){2-7} \cmidrule(lr){8-13}
\textbf{Method} & {10} & {50} & {100} & {250} & {500} & {1000} & {10} & {50} & {100} & {250} & {500} & {1000} \\
\midrule
Flat            & 0.2440 & 0.4917 & 0.6155 & 0.7700 & 0.8613 & 0.9247 & 5.61 & 7.85 & 9.11 & 11.28 & 13.37 & 16.01 \\
IVFFlat         & 0.2439 & 0.4910 & 0.6138 & 0.7661 & 0.8545 & 0.9135 & 6.22 & 9.40 & 11.40 & 15.11 & 18.97 & 23.52 \\
\textsc{dinosaur} ($S=2$) & 0.2441 & 0.4913 & 0.6142 & 0.7666 & 0.8546 & 0.9126 & 6.08 & 9.01 & 10.83 & 13.94 & 17.87 & 41.97 \\
\textsc{dinosaur} ($S=3$) & 0.2439 & 0.4903 & 0.6130 & 0.7648 & 0.8521 & 0.9095 & 6.05 & 8.92 & 10.69 & 13.73 & 18.57 & 52.76 \\
\textsc{dinosaur} ($S=4$) & 0.2440 & 0.4908 & 0.6133 & 0.7649 & 0.8522 & 0.9096 & 6.06 & 8.89 & 10.65 & 13.77 & 19.47 & 59.18 \\
\textsc{dinosaur} ($S=5$) & 0.2438 & 0.4905 & 0.6130 & 0.7646 & 0.8521 & 0.9094 & 6.02 & 8.89 & 10.59 & 13.84 & 20.37 & 62.97 \\
\bottomrule
\end{tabular}
\end{table*}

\begin{table*}[t]
\vspace{-1ex}
\centering
\caption{Stratified Catalogue Coverage and Inclusion Rate at $K{=}1000$ by popularity quintile. Higher is better.}
\label{tab:stratified}
\begin{tabular}{l SSSSS SSSSS}
\toprule
& \multicolumn{5}{c}{\textbf{Catalogue Coverage (\%)}} & \multicolumn{5}{c}{\textbf{Inclusion Rate (\%)}} \\
\cmidrule(lr){2-6} \cmidrule(lr){7-11}
\textbf{Method} & {B0} & {B1} & {B2} & {B3} & {B4} & {B0} & {B1} & {B2} & {B3} & {B4} \\
\midrule
Flat            & 77.63 &  2.66 &  0.00 &  0.00 &  0.00 & 100.00 &  1.11 &  0.00 &  0.00 &  0.00 \\
IVFFlat         & 95.37 & 22.65 &  0.20 &  0.00 &  0.00 & 100.00 &  2.27 &  0.00 &  0.00 &  0.00 \\
\textsc{dinosaur} ($S=2$) & 95.59 & 27.26 & 17.00 & 28.83 & 39.01 & 100.00 & 11.07 &  8.59 & 10.19 & 11.35 \\
\textsc{dinosaur} ($S=3$) & 96.24 & 29.52 & 25.28 & 46.31 & 62.43 & 100.00 & 12.36 &  9.76 & 11.70 & 12.77 \\
\textsc{dinosaur} ($S=4$) & 96.70 & 31.77 & 30.28 & 57.28 & 74.92 & 100.00 & 13.23 & 11.06 & 12.39 & 13.60 \\
\textsc{dinosaur} ($S=5$) & 96.83 & 32.43 & 34.57 & 63.81 & 81.88 & 100.00 & 13.43 & 11.21 & 12.51 & 13.76 \\
\bottomrule
\end{tabular}
\end{table*}

To address \textbf{RQ1} and \textbf{RQ2}, we examine the relationship between global catalogue coverage and recall across varying sample limits $S$. As demonstrated in Table \ref{tab:recall_coverage}, \textsc{dinosaur} profoundly expands the catalogue footprint. For $K=1000$, increasing the sampling budget from the baseline ($S=1$) to $S=5$ elevates global coverage from $23.52\%$ to $62.97\%$. Crucially, this near-tripling of the exploration surface area comes at virtually no cost to point-estimate utility, suffering only a $0.0041$ reduction in offline recall.

This dynamic becomes even more apparent when inspecting the stratified results in Table \ref{tab:stratified}. The baseline approximate search completely ignores the least popular items, achieving $0.00\%$ coverage in bucket B4. By contrast, \textsc{dinosaur} successfully retrieves over $81.88\%$ of these highly obscure items, ensuring that over $13\%$ of users organically receive recommendations from the most extreme end of the long-tail.

Finally, we validate the efficiency of this trade-off by comparing \textsc{dinosaur} against a strong $\epsilon$-greedy heuristic. As shown in Figure \ref{fig:pareto}, \textsc{dinosaur} strictly Pareto-dominates uniform random exploration. At any given fraction of modified candidate lists, and across multiple practical retrieval depths ($K \in \{100, 500, 1000\}$), Thompson Sampling via synthetic posteriors consistently preserves higher top-tier recall while injecting relevant diversity.

%% file: 4.Conclusions.tex
\section{Conclusions \& Future Work}
\label{sec:conclusion}
We introduce \textsc{dinosaur}, a lightweight framework for uncertainty-aware candidate generation with existing ANN infrastructure. Rather than changing the scoring model or retrieval backend, \textsc{dinosaur} indexes multiple samples from item embedding distributions and deduplicates retrieved neighbours back to item identifiers. This simple mechanism turns embedding uncertainty into retrieval-stage exploration, increasing the chance that uncertain or long-tail items enter the candidate set before downstream ranking.

Our experiments show that this can substantially expand catalogue coverage with only small losses in offline recall, and that the resulting recall--exploration trade-off compares favourably to uniform $\epsilon$-greedy exploration. These results suggest that uncertainty-aware retrieval is a practical way to improve long-tail exposure at the candidate-generation stage, where deterministic point-estimate ANN search can otherwise impose hard geometric cut-offs.

Several limitations remain.
First, our experiments instantiate uncertainty through an interaction-count-based proxy; the same retrieval mechanism could instead consume learned uncertainty estimates.
Second, offline recommendation logs are MNAR, so recall on historical interactions measures agreement with previously exposed items rather than the full utility of newly surfaced candidates.
Future work should therefore combine propensity-aware offline evaluation, such as IPS or doubly robust estimators, with online A/B testing to measure the user- and marketplace-level impact of uncertainty-aware retrieval.
We hope our work can inspire such validation of \textsc{dinosaur} in broader use-cases and practical instantiations of ANN-based candidate retrieval.\looseness=-1